\documentclass[conference,letterpaper,final]{IEEEtran}

\usepackage{epsfig}
\usepackage{graphicx}
\usepackage{amsmath}
\usepackage{amsfonts}
\usepackage{amssymb}
\usepackage{psfrag}
\usepackage{subfigure}
\usepackage{booktabs}
\usepackage{flushend}

\newcommand{\snr}{{\sf snr}}

\newcommand{\Hm}{{\boldsymbol H}}
\newcommand{\Ym}{{\boldsymbol Y}}
\newcommand{\Xm}{{\boldsymbol X}}
\newcommand{\Zm}{{\boldsymbol Z}}

\newcommand{\sv}{{\boldsymbol s}}

\newcommand{\Ec}{{\mathcal E}}
\newcommand{\Xc}{{\mathcal X}}
\newcommand{\Yc}{{\mathcal Y}}

\newcommand{\CC}{\mathbb{C}}
\newcommand{\EE}{\mathbb{E}}
\newcommand{\RR}{\mathbb{R}}

\DeclareMathOperator{\trace}{{\rm tr}}

\newcommand{\figurename}{Figure}
\newcommand{\figref}[1]{\figurename~\ref{#1}}

\newcommand{\Rs}{R_{\rm s}}
\newcommand{\Rc}{R_{\rm c}}

\newcommand{\nt}{n_{\rm t}}
\newcommand{\nr}{n_{\rm r}}

\newcommand{\openone}{\leavevmode\hbox{\small1\normalsize\kern-.33em1}}

\newtheorem{theorem}{Theorem}

\newtheorem{remark}{Remark}

\title{Distortion Outage Probability in MIMO Block-Fading Channels}
\author{
\authorblockN{Li Peng}
\authorblockA{Department of Engineering\\ 
University of Cambridge\\
Cambridge CB2 1PZ, UK\\
\tt lp327@cam.ac.uk}
\and
\authorblockN{Albert Guill{\'e}n i F{\`a}bregas}
\authorblockA{Department of Engineering\\
University of Cambridge \\
Cambridge, CB2 1PZ, UK \\
\tt albert.guillen@eng.cam.ac.uk}
}

\begin{document}

\maketitle

%\footnotetext[1]{This work has been partly supported by the China Scholarship Council.}
%\setcounter{footnote}{1}
%\vspace{-18mm}
%%%%%%%%%%%%%%%%%%%%%%%%%%%%%%%%%
\begin{abstract}
We study analogue source transmission over MIMO block-fading channels with receiver-only channel state information. Unlike previous work which considers the end-to-end expected distortion as a figure of merit, we study the distortion outage probability. We first consider the well known transmitter informed bound, which yields a benchmark lower bound to the distortion outage probability of any coding scheme. We next compare the results with source-channel separation. The key difference from the expected distortion approach is that if the channel code rate is chosen appropriately, source-channel separation can not only achieve the same diversity exponent, but also the same distortion outage probability as the transmitter informed lower bound. 
\end{abstract}
%%%%%%%%%%%%%%%%%%%%%%%%%%%%%%%%%

%\newpage

%%%%%%%%%%%%%%%%%%%%%%%%%%%%%%%%%
\section{Introduction}
\label{sec:Introduction}
%%%%%%%%%%%%%%%%%%%%%%%%%%%%%%%%%
The block-fading channel was introduced in \cite{ozarow1994} in order to model delay-limited transmission over slowly varying wireless communications channels. In the channel, each codeword spans only a finite and fixed number $N$ of independent fading blocks. Practical scenarios include OFDM and frequency hopping for low-mobility wireless scenarios. Under this setup, it follows that the Shannon capacity of this channel is zero since there is an irreducible probability that a given transmission rate $\Rc$ is not supported by a particular channel realisation \cite{ozarow1994,biglieri1998}. In particular, a communication outage occurs whenever the instantaneous mutual information is less than the target data rate we wish to communicate at \cite{ozarow1994,biglieri1998}. As shown in \cite{malkamaki1999cdb}, the outage probability is the natural fundamental limit of the channel. 
%Multiple antennas at the transceivers are a key technology that can greatly improve the reliability and data rates over wireless channels \cite{foschini1998lwc,telatar1999cma}. In particular, multiple antennas are known to reduce the outage probability in block-fading channels. 
An important reliability metric over block-fading channels is the SNR exponent or outage diversity, defined as the high-SNR slope of the outage probability in a log-log scale.

Inspired by the work by Laneman {\em et al.} \cite{laneman2005}, the end-to-end expected distortion has been studied to characterise the performance of continuous or analogue source transmission over outage-limited multiple-antenna fading channels \cite{caire2007,deniz2007,deniz2008,tim2008}. The above works consider the SNR exponent of the end-to-end expected distortion (where the expectation is also taken over the fading) as a performance metric for a number of joint source-channel coding schemes. In particular, when the expected distortion is considered, these references illustrate the suboptimality of source-channel separation. In order to improve the performance, a number of joint-source channel schemes based on hybrid analogue-digital or multi-layered coding have been proposed \cite{caire2007,deniz2007,deniz2008,tim2008}. 

The expected distortion is the natural performance metric for {\em ergodic} fading channels, or channels with no stringent delay constraints. However, for outage-limited channels or channels with stringent delay constraints, the expected distortion fails to characterise the true end-to-end performance of such wireless systems, and thus it might not be the appropriate performance metric. In this paper, we take a different approach to the same problem. In particular, our contribution is the study of the distortion outage probability instead of the expected distortion as a figure of merit for system performance. The distortion outage probability is defined as the probability that the instantaneous distortion (a random variable that depends on the channel realisation and SNR) is larger than a target quality-of-service (QoS) distortion, i.e., 
\begin{equation}
\label{eq:eq44}
P_{\rm out}( \snr,\bar{D})=\Pr\left\{D(\Hm, \snr)>\bar{D}\right\}
\end{equation}
where $D(\Hm,\snr)$ is the instantaneous distortion achieved at SNR $\snr$ for a given $\nt\times \nr$ multiple-input multiple-output (MIMO) channel realisation $\Hm$, and $\bar{D}$ is the target distortion level characterizing the acceptable QoS of the systems. We are particularly interested in the distortion-outage SNR exponents.
%\begin{equation}
%d_{\rm out} = \lim_{ \snr\to\infty} -\frac{\log P_{\rm out}( \snr,\bar{D})}{\log \snr}.
%\label{eq:eq73}
%\end{equation}

In this work, similarly to previous works \cite{laneman2005,caire2007,deniz2007,deniz2008,tim2008} we first study a lower bound on the distortion outage performance, i.e., the transmitter informed bound. This bound assumes perfect channel state information at the transmitter (CSIT), and allows to adapt the joint source-channel code to the instantaneous channel conditions. In particular, we find out the relationships between this lower bound on the distortion outage probability and the information outage probability \cite{ozarow1994,biglieri1998}, as well as the corresponding SNR exponents.

%We show that the SNR exponent of the transmitter informed bound is
%\begin{equation}
%d_{\rm out}^{\rm tx}(\Rc(\Hm),b,\bar{D}) = \begin{cases}N\nt \nr&\ \  \text{Gaussian channel inputs} \\
%\nr\left(1+N\left\lfloor \nt-\frac{\Rc}{M}\right\rfloor\right)&\ \ \text{Discrete channel inputs}\end{cases}
%\label{eq:tx_bound_exp_intro}
%\end{equation}
%where $\nt$ is the number of transmit antenna, $\nr$ is the number of receive antenna, N is the number of fading blocks, $\Rc$ is the channel code rate and the size of the constellation $\mathcal{X}$ is $2^M$. 

We next consider source-channel separation \cite{cover}, and show that the separation scheme achieves the same SNR exponent as the transmitter informed bound. We also show that, when the channel coding rate is chosen appropriately, then the separation scheme yields the same distortion outage probability (not only the exponent) as the transmitter informed lower bound. This result, rather surprising a priori --separation is known to be largely suboptimal when the expected distortion SNR exponent is considered \cite{laneman2005,caire2007,deniz2007,deniz2008,tim2008}-- shows that separation can be optimal, when the distortion outage probability is used as a figure of merit for system design.
%%%%%%%%%%%%%%%%%%%%%%%%%%%%%%%%%%%%%%%%%%%%%%%%%%%%%%%%%%%%%%%%%%%%%%%%%%%%%%%%%%%%%%%%%

%%%%%%%%%%%%%%%%%%%%%%%%%%%%%%%%%%%%%%%%
\section{System Model}
\label{sec:SystemModel}
%%%%%%%%%%%%%%%%%%%%%%%%%%%%%%%%%%%%%%%%
%In this section we describe the MIMO block-fading channel model and joint source-channel code under consideration.

%%%%%%%%%%%%%%%%%%%%%%%%%%%%%%%%%%%%%%%%%%
\subsection{Channel Model}
\label{sec:ChannelModel}
We consider a MIMO block-fading channel model with $N$ fading blocks, $\nt$ transmit and $\nr$ receive antennas, and block  length $L$. The channel model is expressed as
\begin{equation}
\Ym_i=\sqrt{\frac{\snr}{\nt}}\Hm_i\Xm_i+\Zm_i,~~ i=1,\dotsc,N
\label{eq:eq3}
\end{equation}
where $\Hm_i\in\mathbb{C}^{\nr\times \nt}$, $\Xm_i\in\mathbb{C}^{\nt\times L}$, $\Ym_i\in\mathbb{C}^{\nr\times L}$, and $\Zm_i\in \CC^{\nr\times L}$ are the channel matrix, transmitted, received and AWGN signals corresponding to block $i$. We assume that the entries of $\Hm_i$ and $\Zm_i$ are independently circularly symmetric complex Gaussian with zero mean and unit variance $\sim\mathcal{N}_{\mathbb{C}}(0,1)$. We define the space-time codewords as $\Xm=\left[\Xm_1,\ldots,\Xm_N\right]$, and we assume they are normalized in energy, i.e., satisfying $\frac{1}{\nt NL}\trace\left(\EE\left[\Xm^H\Xm\right]\right)\leq 1$. The input and output alphabets are denoted by $\Xc^{\nt}$ and $\Yc^{\nr}$, respectively. We consider both random codes constructed using Gaussian and discrete channel inputs (PSK, QAM). For discrete channel inputs we define $m=\log_2|\Xc|$.

We define $\Hm=\text{diag}(\Hm_1,\ldots,\Hm_N)$, assumed to be known perfectly to the receiver. For simplicity we assume that the entries of $\Hm_i$ are i.i.d. $\sim\mathcal{N}_{\mathbb{C}}(0,1)$ (Rayleigh fading), so that $\frac{1}{\nt \nr}\trace\left(\EE\left[\Hm_i^H\Hm_i\right]\right)\leq 1$ and the average SNR per receive antenna is $\snr$. We assume that the transmitter knows the statistics of the channel, but not the channel realisation. Let 
\begin{align}
&I_{\Hm}(\snr)%=\mathbb{E}_{\Xm,\Ym|\Hm}\left[\log_2\frac{P_{\Ym|\Xm,\Hm}(\Ym| \Xm,\Hm)}{P_{\Ym|\Hm}(\Ym|\Hm)}\right]\\
&=\frac{1}{N}\sum_{i=1}^N\mathbb{E}\left[\log_2\frac{P_{\Ym_i|\Xm_i,\Hm_i}(\Ym_i| \Xm_i,\Hm_i)}{P_{\Ym_i|\Hm_i}(\Ym_i|\Hm_i)}\biggl | \Hm_i\right]
\end{align}
denote the instantaneous mutual information of the channel, for a given channel realisation $\Hm$.
%%%%%%%%%%%%%%%%%%%%%%%%%%%%%%%%%%%%%%%%%%%%%%%%%%%%%%%%%%%%%%%%%%%%%%%%%%%%%%%%%%%%%%%%%

%%%%%%%%%%%%%%%%%%%%%%%%%%%%%%%%%%%%%%%%%%
\subsection{Joint Source-Channel Coding}
\label{sec:JointSourceChannelCode}
We consider transmission of analogue sources over the MIMO block-fading channel described in Section \ref{sec:ChannelModel}. Consider a real continuous source that outputs $K$-length vector $\sv\in\mathbb{R}^K$. A $K$-to-$(\nt\times NL)$ joint source-channel encoder is a mapping $\phi:\mathbb{R}^K\rightarrow\mathbb{C}^{\nt\times NL}$ that maps blocks of $K$ source symbols $\sv\in\mathbb{R}^K$ onto length $NL$ space-time channel codewords $\Xm=\left[\Xm_1,\ldots,\Xm_N\right]$. At the receiver end, the corresponding source-channel decoder is a mapping $\varphi:\mathbb{C}^{\nr\times NL}\rightarrow\mathbb{R}^K$ that, for every channel realisation $\Hm$, maps the channel output $\Ym=\left[\Ym_1,\ldots,\Ym_N\right]$ into $\hat{\sv}\in \CC^K$, a reconstruction of the block of source symbols. In order to make explicit the dependencies on $\snr$ and $\Hm$, we denote the reconstructed block of symbols $\hat{\sv}(\snr,\Hm)$. The bandwidth ratio of the code is defined as
\begin{equation}
b\triangleq\frac{NL}{K} ~~\text{channel uses per source symbol}
\end{equation}
The bandwidth ratio can also be expressed as $b=W_c/W_s$, where $W_s,W_c$ are the source and channel bandwidths, respectively. The source rate is denoted by $\Rs=1/W_s$ and the channel rate is denoted by $\Rc=1/W_c$. The average quadratic distortion for a fixed $\Hm$ is given by
\begin{equation}
D(\Hm,\snr)=\frac{1}{K}\mathbb{E}\left[\left|\sv-\hat{\sv}(\snr,\Hm)\right|^2 \bigl|\Hm\right]
\end{equation}
where expectation is with respect to $\sv$, $\hat{\sv}$ and the channel noise, but depends on $\snr$ and on the channel realisation $\Hm$.
Mirroring results from channel coding for block-fading channels \cite{ozarow1994,biglieri1998,malkamaki1999cdb}, we define distortion outage probability
\begin{equation}
P_{\rm out}( \snr,\bar{D})\triangleq\Pr\left\{D(\Hm, \snr)>\bar{D}\right\}.
\label{eq:pout_def}
\end{equation}
We consider a family of joint source-channel coding schemes $\left\{\mathcal{C}_b\right\}$ of bandwidth ratio $b$. The distortion outage probability SNR exponent of the family is defined as
\begin{equation}
d_{\rm out}^*(b,\bar{D})=\sup_{\mathcal{C}_b}\lim_{\snr\rightarrow\infty}\frac{-\log P_{\rm out}( \snr,\bar{D})}{\log\snr}.
\end{equation}
in the forthcoming sections we study the distortion outage probability and the corresponding SNR exponents, and we compare them to those obtained using the expected distortion as a figure-of-merit.

%%%%%%%%%%%%%%%%%%%%%%%%%%%%%%%%%%%%%%%%%%%%%%%%%%%%%%%%%%%%%%%%%%%%%%%%%

%%%%%%%%%%%%%%%%%%%%%%%%%%%%%%%%%%%%%%%%%%
\section{Informed Transmitter}
\label{sec:InformedTransmitter}
%%%%%%%%%%%%%%%%%%%%%%%%%%%%%%%%%%%%%%%%%%
We now study the distortion outage probability exponent for the transmitter informed bound which assumes availability of channel state information at the transmitter (CSIT). Hence, the transmitter can choose the coding rate $\Rc(\Hm)$ equal to the instantaneous mutual information of the $N$-block MIMO fading channel, and the source rate $\Rs=\Rc(\Hm)b$. As shown in \cite{laneman2005,caire2007,deniz2007,deniz2008,tim2008}, this scheme is pointwise optimal for each $\Hm$, and its distortion outage probability is a lower bound on the minimum achievable distortion outage probability for any system of the original channel. In particular, the transmitter informed bound selects the channel coding rate $\Rc(\Hm)=I_{\Hm}(\snr)$. 
Then, the instantaneous end-to-end distortion for a given channel realisation with a Gaussian source of unit variance and an informed transmitter is
\begin{equation}
D(\Hm,\snr) = 2^{-2bI_{\Hm}(\snr)}.
\label{eq:eq8}
\end{equation}
Substituting Equation \eqref{eq:eq8} into Equation \eqref{eq:pout_def}, we can write the transmitter informed bound on the distortion outage probability as
\begin{align}
P_{\rm out}(\snr,\bar{D})%&=\Pr\left\{D(\Hm,\snr)>\bar{D}\right\}\nonumber\\
%&=\Pr\left\{2^{-2b I_{\Hm}(\snr)}>\bar{D}\right\}\nonumber\\
%&=\Pr\left\{-\frac{2b}{N}\inf_{\mathbf{Q}_n\succ0,\ \sum_{n=1}^N\textbf{tr}(\mathbf{Q}_n)\leq N\nt}\sum_{n=1}^N\log_2\det\left(\mathbf{I}+\frac{\snr}{\nt}\Hm_n\mathbf{Q}_n\Hm_n^H\right)>\log_2\bar{D}\right\}\nonumber\\
&=\Pr\left\{I_{\Hm}(\snr)<-\frac{\log_2\bar{D}}{2b}\right\}
\label{eq:eq47}
\end{align}
which shows that the transmitter informed bound on the distortion outage probability can be written as the information outage probability $\Pr\{I_\Hm(\snr)<R\}$ \cite{ozarow1994,biglieri1998} evaluated at target rate $R=-\log_2(\bar{D})/2b$.
We next examine the behavior of the SNR exponent. Following closely the arguments in \cite{lizhong2003}, we have the following result.
\begin{theorem}
\label{thm:thm1}
The SNR exponents of the transmitter informed lower bound for any fixed bandwidth ratio $b>0$, any fixed target distortion level $0\leq \bar{D}\leq 1$ are given by
\begin{equation}
d_{\rm out}^{\rm tx}(b,\bar{D}) = \nt\nr N
\label{eq:tx_bound_exp_infor_gi}
\end{equation}
for Gaussian channel inputs, while for discrete channel inputs
\begin{equation}
d_{\rm out}^{\rm tx}(b,\bar{D})=
\nr\left(1+\left\lfloor N\left(\nt-\frac{\Rs(\bar{D})}{b\,m}\right)\right\rfloor\right)
\label{eq:tx_bound_exp_infor_di}
\end{equation}
where $\Rs(\bar{D})\triangleq-\frac{\log_2(\bar{D})}{2}$ is the rate-distortion region of the source evaluated at $\bar D$.
\end{theorem}
\begin{proof}
The transmitter informed lower bound on the distortion outage probability can be written as the information outage probability \cite{ozarow1994,biglieri1998} evaluated at $R=-\frac{\log_2(\bar{D})}{2b}$ (see Eq. \eqref{eq:eq47}). For Gaussian inputs, the SNR exponent of the information outage probability is $\nt\nr N$ for $R>0$. Since $0\leq \bar{D}\leq 1$,   
$-\frac{\log_2(\bar{D})}{2b}$ is positive. Then,  the resulting SNR exponent is also $\nt\nr N$. For discrete inputs we have that the SNR exponent of the information outage probability is given by the Singleton bound $\nr\left(1+\left\lfloor N\left( \nt-\frac{R}{m}\right)\right\rfloor\right)$ for $0\leq R\leq m$ \cite{NguyenThesis2009,nguyen2009}. Then, the resulting SNR exponent with discrete inputs is given by $\nr\left(1+\left\lfloor N\left( \nt-\frac{R_{\rm s}(\bar{D})}{b\,m}\right)\right\rfloor\right)$ for $R_{\rm c}(\bar{D})=-\frac{\log_2(\bar{D})}{2b}\in\left[0,m\right]$. 
\end{proof}

It is important to note that, since the  transmitter informed lower bound on the distortion outage probability has the exponents given by Theorem \ref{thm:thm1}, the SNR exponents of any coding scheme will be upper bounded by Eqs. \eqref{eq:tx_bound_exp_infor_gi} and \eqref{eq:tx_bound_exp_infor_di}.

%The converse of the informed transmitter lower bound is easily proved by invoking Data Processing Inequality. It is obvious that the encoding and decoding process $\sv\rightarrow\Xm\rightarrow\Ym\rightarrow\hat{\sv}$ can be treated as a Markov Chain, by applying Data Processing Inequality, we have $I(\sv,\hat{\sv})\leq I(\Xm,\Ym)$. And also noticing $I(\sv,\hat{\sv})$ and $I(\Xm,\Ym)$ are in different units, bits/source symbol for $I(\sv,\hat{\sv})$ and bits/channel use for $I(\Xm,\Ym)$, and by invoking the rate-distortion function of Gaussian source,
%\begin{equation}
%\frac{1}{2}\log\frac{1}{D(\Hm,\snr)}\leq bI(\Xm,\Ym)
%\end{equation}
%Hence we can show any distortion outage probability lower than \eqref{eq:eq47} is not possible.

\begin{remark}[Diversity-Multiplexing Tradeoff]
\label{rem:dmt}
The results of Theorem \ref{thm:thm1} for Gaussian channel inputs can be generalized to a family of joint source and channel codes whose rate increases with $\snr$. In particular, letting $\Rs = r_{\rm s}\log\snr$ and $\Rc = r_{\rm c}\log\snr$ with $b=\frac{ r_{\rm s}}{ r_{\rm c}}$ results in a diversity-multiplexing tradeoff $d_{\rm out}^{\rm tx}(b,\bar{D},r_{\rm c})$ given by the piecewise linear function joining the points ($r_{\rm c},d(r_{\rm c})$)
\begin{equation}
d(r_{\rm c}) = N (\nt-r_{\rm c})(\nr -r_{\rm c})
\end{equation}
achieving the result of Theorem \ref{thm:thm1} for $r_{\rm c}=0$ \cite{lizhong2003}.
\end{remark}

\begin{remark}[Comparison with Expected Distortion]
The SNR exponent of the informed transmitter lower bound with Gaussian channel inputs when the expected distortion is used as performance metric is given by \cite{deniz2008},
\begin{equation}
d_{\rm exp}^{\rm tx}=N\sum_{i=1}^{\min(\nr,\nt)}\min\left\{\frac{2b}{N},2i-1+\left|\nt-\nr\right|\right\}
\end{equation}
In \figref{fig:comp1}, we illustrate the SNR exponents of the transmitter informed bound for Gaussian channel inputs, from both distortion-outage and expected distortion perspectives. We observe that the distortion-outage exponent is always larger (for small bandwidth ratios) or equal to the expected distortion exponent. In \figref{fig:comp2}, we compare the distortion-outage SNR exponents of Gaussian random codes with that of discrete inputs. We observe that full diversity ($N\nr\nt$) is achieved when the we have a large bandwidth ratio for all inputs, and that a larger constellation size results in a larger support with full diversity. Note that Singleton bound is valid when $R_c\leq m$, hence we can obtain a bound on $b$, which is $b\geq-\frac{\log_2\bar{D}}{2m}$. For $b$ smaller than this threshold, the exponent is zero.%{\bf ADD COMMENT HERE ON BOUND ON $b$. THE EXPONENT SHOULD BE ZERO, RIGHT? 	SAY THAT CONSTELLATIONS OF HIGHER-ORDER ALPHABETS ACHIEVE LARGER EXPONENTS FOR A LARGER RANGE OF $b$.}
%{\bf THE 2 FIGURES SHOULD HAVE Y AXIS UNTIL 35, IT WILL LOOK NICER.}
\begin{figure}
	\centering
		\includegraphics[width=1\columnwidth]{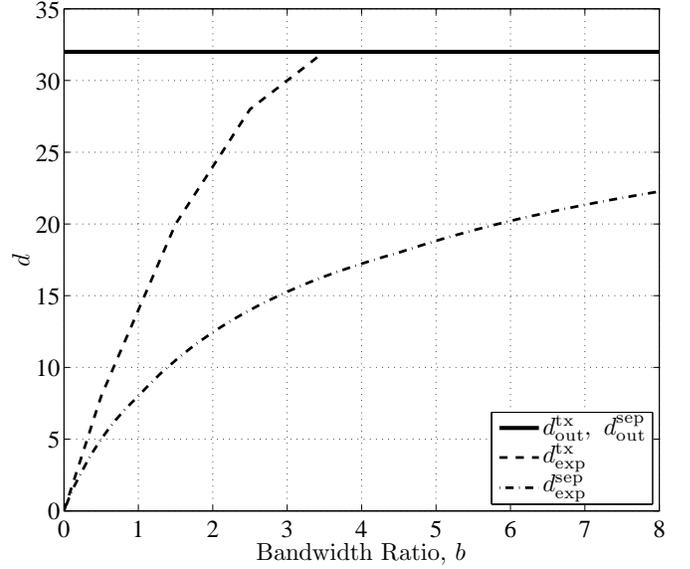}
	\caption{Distortion-outage vs. expected distortion SNR exponents in a $4\times 4$ MIMO block-fading channel with $N=2$, Gaussian inputs and $\bar{D}=0.05$. Solid lines correspond to the distortion-outage exponents, while dashed and dash dotted lines correspond to expected distortion exponents. %The expected distortion SNR exponent achieved by separation is also shown. The distortion-outage SNR exponent achieved by separation is  the same as that of the transmitter informed bound.
	}
	\label{fig:comp1}
 \vspace{-3mm}
 \end{figure}
 
 \begin{figure}
	\centering
		\includegraphics[width=1\columnwidth]{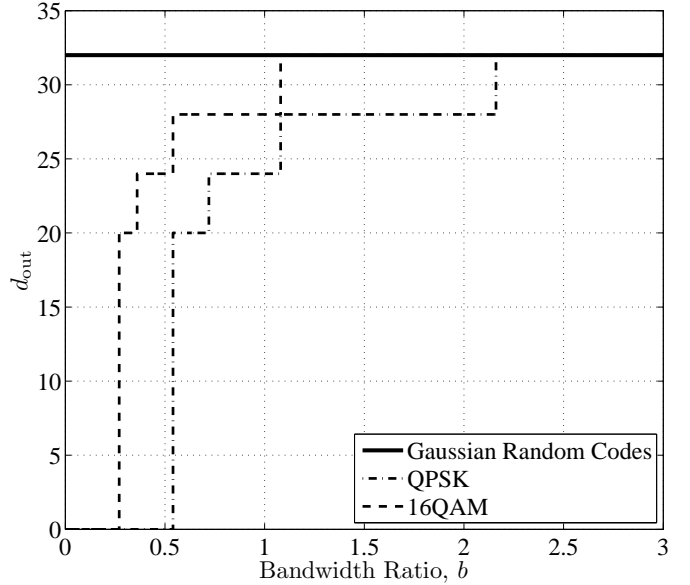}
	\caption{Distortion-outage SNR exponents in a $4\times 4$ MIMO block-fading channel with $N=2$, different channel inputs and $\bar{D}=0.05$. 
	%$b$ is not calculated from zero because $-\frac{\log_2(\bar{D})}{2b}\leq M$. Note that the distortion-outage SNR exponent achieved by separation is the same as that of the transmitter informed bound.
	}
	\label{fig:comp2}
 \vspace{-3mm}
 \end{figure}
\end{remark}
%%%%%%%%%%%%%%%%%%%%%%%%%%%%%%%%%%%%%%%%%%%%%%%%%%%%%%%%%%%%%%%%%%%%%%%

%%%%%%%%%%%%%%%%%%%%%%%%%%%%%%%%%%%%%%%%%
\section{Separation}
\label{sec:Seperation}
%%%%%%%%%%%%%%%%%%%%%%%%%%%%%%%%%%%%%%%%%
A source-channel coding separation scheme consists of the concatenation of a fixed-length block source encoder $\phi_{\rm s}:\RR^K\to \RR^{K}$, of rate $\Rs$ bits per source sample, with a space-time channel encoder $\phi_{\rm c}:\RR^{K}\to \CC^{\nt \times NL}$ of rate $\Rc$ bits per channel use. Source and space-time coding rates are related through the bandwidth ratio as $\Rs=\Rc b$. Let $D_{\rm s}(\Rs)$ denote the distortion-rate function of the source and $P_{\rm e}( \snr,\Hm)$ denote the error probability of the channel code for a particular $\snr$ and channel realisation $\Hm$. Following \cite{hochwald1997}, \cite[Ch. 7]{viterbi1979pdc}, we write the distortion achieved by the separation scheme for a fixed channel realisation $\Hm$ as
\begin{align}
&D_{\rm sep}(\Hm,\snr)\notag\\
 &= D(\snr\,|\,\Hm,\text{no channel error})\Pr\{\text{no channel error}\,|\,\Hm\}\nonumber\\
 &+ D(\snr\,|\,\text{channel error},\Hm)\Pr\{\text{channel error}\,|\,\Hm\}.
 \label{eq:sep_channel_form}
% & \leq D(\Hm,\snr\,|\,\text{no channel error}) \Pr\{\text{no channel error}\} + d_0 P_{\rm e}(\snr,\Hm) 
\end{align} 
Following \cite{hochwald1997}\cite[Ch. 7]{viterbi1979pdc}, we can upperbound \eqref{eq:sep_channel_form} as
\begin{align}
D_{\rm sep}(\Hm,\snr)
% &~\leq D(\Hm,\snr\,|\,\text{no channel error}) + d_0P_{\rm e}( \snr,\Hm)\\
&\leq D_{\rm s}(\Rs)+ d_0P_e( \snr,\Hm)
\label{eq:sep_dist_bound}
\end{align}
where $d_0$ is a bound to the mean MSE distortion and $d_0^2$ upperbounds its variance \cite[Sec. 7.5]{viterbi1979pdc}. 
%The full derivation of \eqref{eq:sep_dist_bound} for discrete memoryless sources is given in \cite[Theorem 7.3.1]{viterbi1979pdc}. 
%The proof for the Gaussian source is given in the Appendix. 
%The proof of \eqref{eq:sep_dist_bound} is omitted or the sake of space limitations. 
Since the channel realisation is unknown to the transmitter, the average distortion when there is no channel error is the distortion-rate function of the source code, and that it does not depend on $\Hm$ nor $\snr$.
Using Gallager's error exponents for channel coding, we further upperbound \eqref{eq:sep_dist_bound} as \cite{gallager1968ita}
\begin{equation}
D_{\rm sep}(\Hm,\snr)\leq D_{\rm s}(\Rs)+ d_0 2^{-nE_{\rm r}(\Rc,\Hm)}
\label{eq:d_exp}
\end{equation}
where $n=NL$ is the codeword length,
\begin{equation}
E_{\rm r}(\Rc,\Hm)=\sup_{0\leq\rho\leq1} E_0(\rho,\Hm)-\rho \Rc
\end{equation}
is the random coding error exponent, and $E_0(\rho,\Hm)$
%\begin{align}
%&E_0(\rho,\Hm)=\notag\\
%&\hspace{-2mm}-\log_2\mathbb{E}_{\Xm,\Ym|\Hm}\left[\left(\sum_{\xv'}P_{\Xm}(\xv')\left(\frac{P_{\Ym|\Xm,\Hm}(\Ym| \xv',\Hm)}{P_{\Ym|\Xm,\Hm}(\Ym| \Xm,\Hm)}\right)^{\frac{1}{1+\rho}}\right)^\rho\right]
%\end{align}
is the Gallager function \cite{gallager1968ita}. 
According to Gallager's noisy channel coding theorem, the random coding error exponent $E_{\rm r}(\Rc,\Hm)>0$ whenever $\Rc<I_{\Hm}(\snr)$ and $E_{\rm r}(\Rc,\Hm)=0$ when $\Rc\geq I_{\Hm}(\snr)$ \cite{gallager1968ita}. 

It is also clear from \eqref{eq:d_exp}, that any separation scheme will achieve a distortion that is upperbounded by
\begin{equation}
D_{\rm sep}(\Hm,\snr)< D_{\rm s}(\Rs)+ d_0.
\label{eq:d_upbnd}
\end{equation}
Furthermore, for large $n$ we have that
\begin{equation}
\lim_{n\rightarrow\infty} P_{\rm e}(\snr,\Hm)=\begin{cases}1 &\Rc\geq I_{\Hm}(\snr) \\
0 &\Rc<I_{\Hm}(\snr)\end{cases}
\end{equation}
Therefore, in the limit for large $n$, we obtain that the distortion obtained with separation can be upper bounded as
\begin{equation}
D_{\rm sep}(\Hm,\snr)\leq D_{\rm s}(\Rs)+ d_0\openone\left\{I_{\Hm}(\snr)\leq \Rc\right\}.
\label{eq:eq78}
\end{equation}
where $\openone\{\Ec\}$ is the indicator function of the event $\Ec$. 

The corresponding distortion outage probability is therefore simply bounded as 
\begin{align}
P_{\rm out}^{\rm sep}&(\snr,\bar{D})=\Pr\left\{D_{\rm sep}(\Hm,\snr)>\bar{D}\right\}\nonumber\\
&\leq \Pr\Bigl\{D_{\rm s}(b\Rc)+ d_0\openone{\left\{I_{\Hm}(\Xm;\Ym)\leq \Rc\right\}}>\bar{D}\Bigr\}.
\label{eq:eq79}
\end{align}

We have the following result.
\begin{theorem}
\label{thm:thm2}
The distortion outage SNR exponent of a tandem separation scheme is given by  
\begin{equation}
d_{\rm out}^{\rm sep}(\Rc,b,\bar{D})=N\nt\nr
\label{eq:eq_sep_exp_gi}
\end{equation} 
for Gaussian channel inputs, while for discrete channel inputs
\begin{equation}
d_{\rm out}^{\rm sep}(\Rc,b,\bar{D})=
\nr\left(1+\left\lfloor N\left(\nt-\frac{\Rc}{m}\right)\right\rfloor\right).
\label{eq:eq_sep_exp_di}
\end{equation} 
\end{theorem}
\begin{proof}
From \eqref{eq:eq79} we have
\begin{align}
P_{\rm out}^{\rm sep}(\snr,\bar{D})&\leq\Pr\Bigl\{D_{\rm s}(b\Rc)+d_0\openone\left\{I_{\Hm}(\snr)\leq \Rc\right\}>\bar{D}\Bigr\}\nonumber\\
&=\Pr\left\{\openone\left\{I_{\Hm}(\snr)\leq \Rc\right\}>\frac{\bar{D}-D_{\rm s}(b\Rc)}{d_0}\right\}.
\label{eq:eq80}
\end{align}
We note that the quantity $\frac{\bar{D}-D_{\rm s}(b\Rc)}{d_0}\in[0,1)$, since
\begin{align}
D_{\rm s}(b\Rc)\leq \bar{D}< D_{\rm s}(b\Rc)+ d_0.
\label{eq:eq_d_range}
\end{align}
Then, since the  indicator function takes only the values 0 or 1, we rewrite \eqref{eq:eq80} as
\begin{align}
P_{\rm out}^{\rm sep}(\snr,\bar{D})\leq&\Pr\left\{\openone\left\{I_{\Hm}(\snr)\leq \Rc\right\}>\frac{\bar{D}-D_{\rm s}(b\Rc)}{d_0}\right\}\nonumber\\
&=\Pr\left\{I_{\Hm}(\snr)\leq \Rc\right\},
\label{eq:eq_sep_outage}
\end{align}
which is exactly the information outage probability of the MIMO block fading channel when the channel coding rate equals to $\Rc$. Hence, the result follows from \cite{lizhong2003,NguyenThesis2009,nguyen2009}.
\end{proof}

\begin{remark}
From the above result and Remark \ref{rem:dmt} it is clear that the same diversity-multiplexing tradeoff will be achieved in the case of separation as well.
\end{remark}

\indent From \eqref{eq:eq_d_range}, we find that for each $\bar{D}$, there is a range of values for coding rate $\Rc$ that we can use to achieve the exponents in \eqref{eq:eq_sep_exp_gi} and \eqref{eq:eq_sep_exp_di}.
\begin{align}
\bar{D}-d_0<&2^{-2b\Rc}\leq\bar{D}
\end{align}
which readily implies $\frac{1}{b}\Rs(\bar{D})\leq \Rc<\frac{1}{b}\Rs(\bar{D}-d_0)$,
where $\frac{1}{b}\Rs(\bar{D})=-\frac{\log_2\bar{D}}{2b}$ for the real Gaussian source.
Equation \eqref{eq:eq_sep_outage} implies that when $\Rc=\frac{1}{b}\Rs(\bar{D})$, the distortion outage probability for separation can be upper bounded by
\begin{equation}
P_{\rm out}^{\rm sep}(\snr, \bar{D})\leq\Pr\left\{I_{\Hm}(\snr)\leq-\frac{\log_2\bar{D}}{2b}\right\}
\end{equation}
which coincides with the transmitter informed bound, and hence achieves the minimum possible distortion outage probability making separation with this particular choice of the channel coding rate distortion-outage optimal.

\begin{remark}[Comparison with Expected Distortion]
The expected distortion exponent for separation scheme for $N$ block-fading channel is given by \cite{caire2007}, for $\frac{1}{b}\in\left[\frac{2(j-1)}{d^*(j-1)},\frac{2j}{d^*(j)}\right)$
\begin{align}
d_{\rm exp}^{\rm sep}=N\frac{2b(jd^*(j-1)-(j-1)d^*(j))}{2b+d^*(j-1)-d^*(j)}
\end{align}
for $j=1,\ldots,\min(\nr,\nt)$. Where $d^*(k)$ is the optimal tradeoff curve given by $d^*(k)=N(\nr-k)(\nt-k)$.
%{\bf WHAT IS $d^*$?????}
\end{remark}
We observe from \figref{fig:comp1} that the distortion-outage exponent of separation exhibits a large gain over its expected distortion exponent counterpart for all bandwidth ratios. Remark that Theorem 2 shows that the distortion outage exponents are equal to those of the transmitter informed bound. 

%furthermore, as it is shown in the next figure, separation not only achieves the optimal distortion-outage exponent, but achieves the actual transmitter informed bound.

\begin{figure}
	\centering
		\includegraphics[width=1\columnwidth]{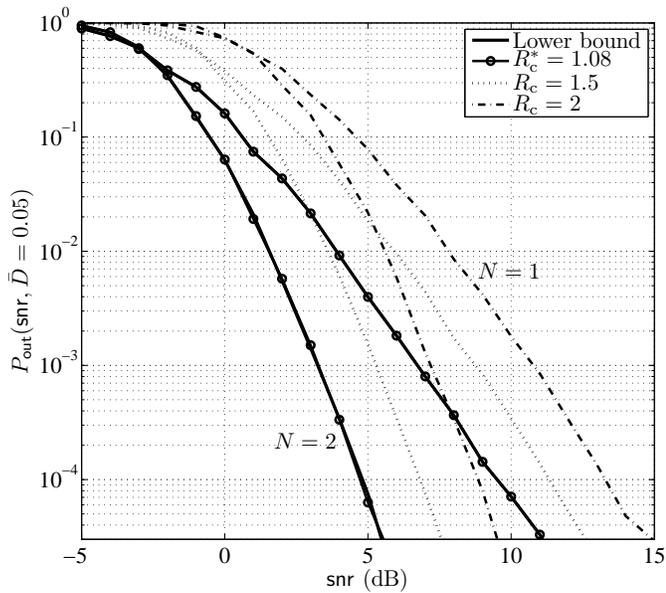}
	\caption{Transmitter informed lower bound and separation upper bound on the distortion outage probability with Gaussian source and channel inputs, $b=2$, $d_0=0.5$ in a $2\times2$ MIMO system. In this case, $\Rc^*=\frac{1}{b}\Rs(\bar{D})=1.08$.}
	\label{fig:mimo2by2_5}
 \vspace{-3mm}
 \end{figure}

\begin{figure}[htb]
 \centering
 \includegraphics[width=1\columnwidth]{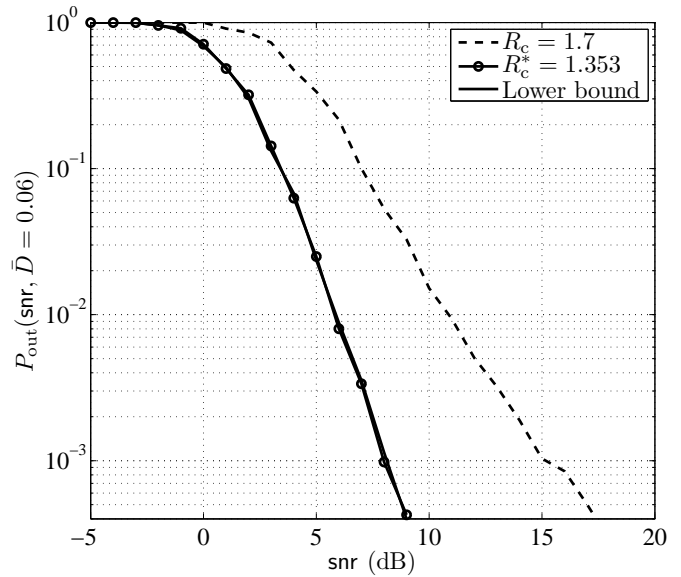}		
 \caption{Transmitter informed lower bound and separation upper bound on the distortion outage probability with Gaussian source and BPSK channel inputs in a $2\times2$ system with $N=2$, $b=1.5$ and $d_0=0.5$. In this case, $\Rc^*=\frac{1}{b} \Rs(\bar D) = 1.353$.  %{\bf the SNR in the ylabel needs to be in sf font.}
 }
 \label{fig:mimo2by2_6}
 \vspace{-3mm}
\end{figure}

As examples, we show in \figref{fig:mimo2by2_5} the distortion outage probability of a $2\times2$ MIMO block-fading channel with i.i.d. Rayleigh fading with $N=1,2$ and $\bar D = 0.05$, for Gaussian source and channel inputs. As predicted by our results, the transmitter informed lower bound for distortion outage probability using informed transmitter always has a slope that equals to $N\nt\nr$ and it is independent of channel coding rate $\Rc$, target distortion $\bar{D}$ and bandwidth ratio $b$. The figure validates our results that the SNR exponent of distortion outage probability of separation scheme is equal to $N\nt\nr$. We also note that when the channel coding rate is chosen to be $\Rc=\frac{1}{b}R_{\mathcal{Q}}(\bar{D})$, the resulting distortion outage probability upperbound matches the transmitter informed bound. We also have shown in \figref{fig:mimo2by2_6} the distortion outage probability of a $2\times2$ MIMO block-fading channel with $N=2$ for $\bar{D}=0.06$ with BPSK. We observe an exponent of 4 when $\Rc^*=1.353$, and 2 when $\Rc=1.7$, as predicted by the  Singleton bound. We remark that for high $\Rc$, there is a significant loss in distortion outage probability (not only in gain, but also in exponent) due to the Singleton bound.

%%%%%%%%%%%%%%%%%%%%%%%%%%%%%%%%%%%%%%%%%
\section{Conclusions}
\label{sec:Conclusions}
%%%%%%%%%%%%%%%%%%%%%%%%%%%%%%%%%%%%%%%%%%%

We have revisited analogue source transmission over MIMO block-fading channels and proposed the distortion outage probability as a new performance metric for system design. We have argued that the distortion outage probability is the natural performance metric for delay-limited channels. We have derived the SNR exponents for both Gaussian and coded modulation inputs. We have shown that the distortion-outage SNR exponents are always larger than the expected distortion exponents, in both, transmitter informed bound and separation.  We have furthermore shown that source-channel separation can not only achieve the SNR exponent of the transmitter informed bound, but also the actual distortion outage probability, when the channel coding rate is chosen appropriately.

\newpage
\bibliographystyle{IEEE}
\bibliography{outagedistortion}

\end{document}